\documentclass[12pt]{iopart}

\usepackage{iopams}  

\usepackage{caption}
\usepackage{graphicx}
\usepackage{amsfonts}
\usepackage{amsthm}
\usepackage{color}

\begin{document}

\newtheorem{corollary}{Corollary}[section]
\newtheorem{lemma}{Lemma}[section]
\newtheorem{theorem}{Theorem}[section]
\newtheorem{definition}{Definition}[section]
\newtheorem{paradox}{Paradox}[section]

\title[Stability of the Centrality of Unions of Networks on the Same Vertex Set]{Stability of the Centrality of Unions of Networks on the Same Vertex Set.}

\author{Chuan Wen, Loe$^1$ and Henrik Jeldtoft Jensen$^2$}

\address{Department of Mathematics and Complexity \& Networks Group\\ Imperial College London, London, SW7 2AZ, UK\\}
\ead{(1) c.loe11@imperial.ac.uk,  (2) h.jensen@imperial.ac.uk}
\begin{abstract}
Let $G^1(V,E_1)$ and $G^2(V,E_2)$ be two networks on the same vertex set $V$ and consider the union of edges $G(V, E_1 \cup E_2)$. This paper studies the stability of the Degree, Betweenness and Eigenvector Centrality of the resultant network, $G(V, E_1 \cup E_2)$. Specifically assume $v^1_{max}$ and $v^c_{max}$ are the highest centrality vertices of $G^1(V,E_1)$ and $G(V, E_1 \cup E_2)$ respectively, we want to find $Pr(v^1_{max} = v^c_{max})$.
\end{abstract}

\maketitle

\section{Introduction}
Zachary Karate Club Network represents the consistent interactions of the karate club members outside of classes and club meetings \cite{zachary1977ifm}. ``Outside interactions" were defined by  8 different relationships. One of such relationships is the association in and between academic classes at the university. The 8 different relationships were then combined as edges (of different types) of Zachary Karate Club Network.

In abstraction, each relationship forms a distinct network on the same vertex set (karate club members). The resultant Zachary Karate Club Network combines the edges of the different networks, where the process is known as the \emph{edge union of networks on the same vertex set}.

The degree distribution and clustering coefficient of the resultant network generated by this process were studied in \cite{OurPaper}. The present paper extends the investigation by studying the asymptotic behaviour of the network centrality.

The centrality of a network ranks the vertices according to their importance when information flows through the network. A high centrality vertex is usually inferred as a hub of the network, where its absence could severely decrease the efficiency of communication across the network. For example the top centrality vertex could be a major airport, city or a celebrity in a social network.

There are many forms of social networks like email, blogs, Facebook or Twitter. Assume every social network has a different individual of highest centrality, since individuals have their preferred mode of communication. Thus a celebrity in Twitter most likely does not command the same influence among avid Facebook users.

What is the probability that the highest centrality Twitter user is also placed as a highly central figure in the edge union of Twitter and Facebook network? Since centrality is computationally expensive for large network, it will be helpful to have a theoretical understanding to allow us to estimate how the ranking of the centrality of a vertex changes when networks are unioned together. In short we aim to extrapolate some information about the centrality without recomputing the combined network.

\section{Definitions and Preliminaries}
\begin{definition}
Let $G^1(V,E_1)$ and $G^2(V,E_2)$ be two networks on the same vertex set $V$ and edge set $E_1$ and $E_2$ respectively, where an element in the edge set $e \in E_1$ (or $E_2$) is a vertex pair $e = (u,v)$ and $u,v \in V$. The \textbf{edge union of the networks} gives the \textbf{composite network}, $G^c = G(V, E_1 \cup E_2)$.
\end{definition}

Erd\H{o}s-R\'{e}nyi/Gilbert network is defined as a random network $G_{n,p}$, where there are $n$ vertices and a pair of vertices are connected with probability $p$ \cite{citeulike:4012374,Gilbert1959}. The probability that an edge is in neither $G^1_{n,p}$ nor $G^2_{n,q}$ is $(1-p)(1-q)$. Hence a pair of vertices in the composite network are connected with probability $1-(1-p)(1-q)$ with
\begin{equation}\label{eq:ER2ER}
G^c = G^1_{n,p} \cup G^2_{n,q} \sim G_{n,1-(1-p)(1-q)}
\end{equation}

We can union multiple $G_{n,p}$ networks together, since the composite network is an Erd\H{o}s-R\'{e}nyi network. Hence the edge union of $m$ networks can be expressed as the union of two networks: $G^1$ and $\cup_{2\leq i \leq m}G^i$. Therefore it is sufficient to focus on the results of the edge union of two networks.

\begin{definition}
Let $G^1(V,E_1)$ and $G^2(V,E_2)$ be two networks on the same vertex set $V$. An edge $e \in E_1 \cup E_2$ in the composite network is called a \textbf{common edge} if $e \in E_1 \cap E_2$. Hence the set of common edges is $E_1 \cap E_2$. (see Fig. \ref{graph_example})
\end{definition}

\begin{center}
  \includegraphics[width=130mm]{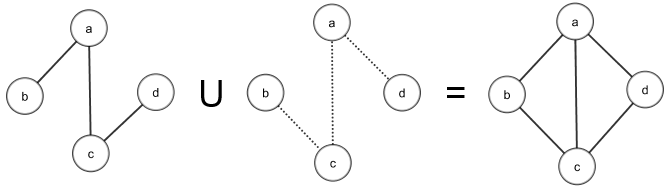}
    \captionof{figure}{Edge union of two networks on the same vertex set. The solid edges and the dotted edges are from the edge set of different networks. The composite network is the rightmost network. Edge $\{a,c\}$ is a common edge between the two networks.}\label{graph_example}
\end{center}

\begin{lemma}\label{lemma:ExpectedCommonEdges}
Let $G^1(V,E_1)$ and $G^2(V,E_2)$ be two networks on the same vertex set $V$. Consider a vertex $v\in V$ where its degree at $G^1$ and $G^2$ is $d_1=deg_1(v)$ and $d_2=deg_2(v)$ respectively. The expected number of common edges between $G^1$ and $G^2$ at vertex $v$ is:
\begin{equation}
\mathcal{C}(v,G^1,G^2) = \frac{d_1 \cdot d_2}{n-1}
\end{equation}
\end{lemma}

Lastly we now use the following notations: $G^1_{n,p}$ and $G^2_{n,q}$ are networks on the same vertex set. Let $v^1_{max}$ and $v^c_{max}$ be the vertices with the highest centrality in $G^1$ and $G^c_{n,p'} = G^1 \cup G^2$ respectively, where $p' = 1 - (1-p)(1-q)$. The degree of vertex $v$ at $G^1$ and $G^c$ are denoted by $deg_1(v)$ and $deg_c(v)$ respectively. We want to study the probability of event $E$ where $v^1_{max} = v^c_{max}$, i.e. $Pr(E) = Pr(v^1_{max} = v^c_{max})$.

%----Degree Centrality----%
\section{Degree Centrality}\label{sec:deg_centrality}
Degree Centrality ranks the vertices according to their degree such that a vertex with a higher degree has higher centrality ranking.

\subsection{Analytical Results}
\begin{theorem}\label{thm:max_deg}
Given $G_{n,p}$, where $0<p<1$. If $\lim_{n \rightarrow \infty}np(1-p)/\ln(n) \rightarrow \infty$ and $x^*$ is a fixed real number, then
\begin{equation*}
\lim_{n \rightarrow \infty} Pr(\mbox{Max Degree} < a^* + \sqrt{np(1-p)}b^*x^*) = e^{-e^{-x^*}}
\end{equation*}
where
\begin{equation*}
a^* = np + \sqrt{2p(1-p)n\ln{n}}\Big(1 - \frac{\ln{\ln{n}}}{4\ln{n}} - \frac{\ln{(2\sqrt{\pi})}}{2\ln{n}} \Big);
\end{equation*}
\begin{equation*}
b^* = \frac{\sqrt{2np(1-p)\ln{n}}}{2\ln{n}}
\end{equation*}
\end{theorem}
\begin{proof}
Proven in \cite{Balbuena}. Sketch: The degree distribution of a random graph can be approximated by a normal distribution $\Phi$ for large values of $n$. Hence the distribution of the maximum degree is $\Phi^n$ from Extreme Value Theory. For appropriate normalizing factors $a^*$ and $b^*$, the standard Gumbel Distribution approximates $\Phi^n$.
\end{proof}

Theorem \ref{thm:max_deg} determines the probability that the maximum degree of a random network is less than some bound $\beta$. If we express the expected value of $deg_c(v^1_{max})$ as the same form as $\beta$, the $Pr(v^1_{max} = v^c_{max})$ follows from Theorem \ref{thm:max_deg}:

\begin{lemma}\label{deg_centrality}
Let $G^1_{n,p}$ and $G^2_{n,q}$ be networks on the same vertex set. Let $v^1_{max}$ and $v^c_{max}$ be the vertices with the highest Degree Centrality of $G^1$ and $G^c = G^1 \cup G^2$ respectively. Then for $n \rightarrow \infty$: 
\begin{equation*}
Pr(v^1_{max} = v^c_{max}) = e^{-e^{-x}},
\end{equation*}
where $x = \frac{\gamma p}{p-pq+q}$ and $\gamma$ is Euler-Mascheroni constant.
\end{lemma}

\begin{proof}
From Theorem \ref{thm:max_deg}, the mean of the standard Gumbel distribution is given by $x^* = \gamma$ and hence the mean maximum-degree of $G^1$ is given by $a^* + \sqrt{np(1-p)}b^*\gamma$. For simplification, we group variable $n$ in each term together, i.e.
\begin{equation}\label{proof:expected_deg1}
\mathbb{E}(deg_1(v^1_{max})) = np + \sqrt{p(1-p)}C_1(n) + p(1-p)C_2(n)\gamma
\end{equation}
where
\begin{equation*}
C_1(n)=\sqrt{2n\ln{n}}\Big(1 - \frac{\ln{\ln{n}}}{4\ln{n}} - \frac{\ln{(2\sqrt{\pi})}}{2\ln{n}} \Big);
\end{equation*}
\begin{equation*}
C_2(n)=\frac{n\sqrt{2\ln{n}}}{2\ln{n}}=\frac{n}{\sqrt{2\ln{n}}}
\end{equation*}

The expected number of edges incident at $v^1_{max}$ in $G^2$  (i.e. $deg_2( v^1_{max}$)) is $nq$, thus the expected number of common edges incident at $v^1_{max}$ is determined by lemma \ref{lemma:ExpectedCommonEdges}:
\begin{eqnarray}
&&(np + \sqrt{p(1-p)}C_1(n) + p(1-p)C_2(n)\gamma)\cdot(nq)/(n-1)\nonumber \\
&\approx& npq + \sqrt{p(1-p)}qC_1(n) + pq(1-p)C_2(n)\gamma \label{proof:expected_ce}
\end{eqnarray}
Let $p' = 1 - (1-p)(1-q)$, the expected degree of $v^1_{max}$ at $G^c$ is Eq. \ref{proof:expected_deg1} + $nq$ - Eq. \ref{proof:expected_ce}:
\begin{eqnarray*}
\mathbb{E}(deg_c(v^1_{max})) &=& (np + \sqrt{p(1-p)}C_1(n) + p(1-p)C_2(n)\gamma) + nq \\
&& - (npq + \sqrt{p(1-p)}qC_1(n) + pq(1-p)C_2(n)\gamma)\\
&=& np' + (1-q)\sqrt{p(1-p)}C_1(n) + p(1-p)(1-q)C_2(n)\gamma
\end{eqnarray*}

We now rearrange the expression for  $\mathbb{E}(deg_c(v^1_{max}))$ to bring it into the same form as Eq. \ref{proof:expected_deg1}, except with the probability $p' = 1 - (1-p)(1-q)$ instead of $p$:
\begin{eqnarray*}
\mathbb{E}(deg_c(v^1_{max})) &=& np' + (1-q)\sqrt{p(1-p)}C_1(n) + p(1-p)(1-q)C_2(n)\gamma\\
&=& np' + \sqrt{p'(1-p')}C_1(n) + p(1-p)(1-q)C_2(n)\gamma\\
&& - (\sqrt{p'(1-p')} - (1-q)\sqrt{p(1-p)})C_1(n)\\
&=& np' + \sqrt{p'(1-p')}C_1(n) + p'(1-p')C_2(n)\Big(\frac{\gamma p}{p-pq+q}\Big)\\
&& - (\sqrt{p'(1-p')} - (1-q)\sqrt{p(1-p)})C_1(n)\\
&=& np' + \sqrt{p'(1-p')}C_1(n) +  p'(1-p')C_2(n)x
\end{eqnarray*}

And rewrite the equation such that the rest of the expression is in $x$
\begin{equation}\label{eq:degC_x}
x = \Big( \frac{\gamma p}{p-pq+q} - \frac{(\sqrt{p'(1-p')} - (1-q)\sqrt{p(1-p)})C_1(n)}{p'(1-p')C_2(n)} \Big)
\end{equation}
Finally $\lim_{n\rightarrow \infty}C_1/C_2 = \lim_{n\rightarrow \infty} O(\ln{n}/\sqrt{n}) = 0$ and with Theorem \ref{thm:max_deg}:
\begin{equation*}
\lim_{n\rightarrow \infty}Pr(\mbox{Max Degree of }G^c < \mathbb{E}(deg_c(v^1_{max}))) = e^{-e^{-x}}.
\end{equation*}
\end{proof}

Lastly in the proof of Theorem \ref{thm:max_deg} \cite{Balbuena}, it stated that the limit approximation works better if $p$ is not to be too large or small. Thus the same condition applies in our case to both p and q for Lemma \ref{deg_centrality}. Therefore since $0 \ll p,q \ll 1$, then $p$ and $q$ have to be within the same order of magnitude.

\subsection{Empirical Results}
\label{empirical}
Figure \ref{fig:deg_centrality} plots the asymptotic behaviour obtained from the simulations. Our  analytical results are describes the asymptotic behaviour of the system in the limit of  large values of $n$ and is accordingly unable to describe smaller networks well (details in section \ref{sec:deg_centrality_asy}). From the empirical results of Figure \ref{fig:deg_centrality}, we observe that the stability of the top Degree Centrality vertex decreases for increasing values of $n$.
\begin{center}
  \includegraphics[width=130mm]{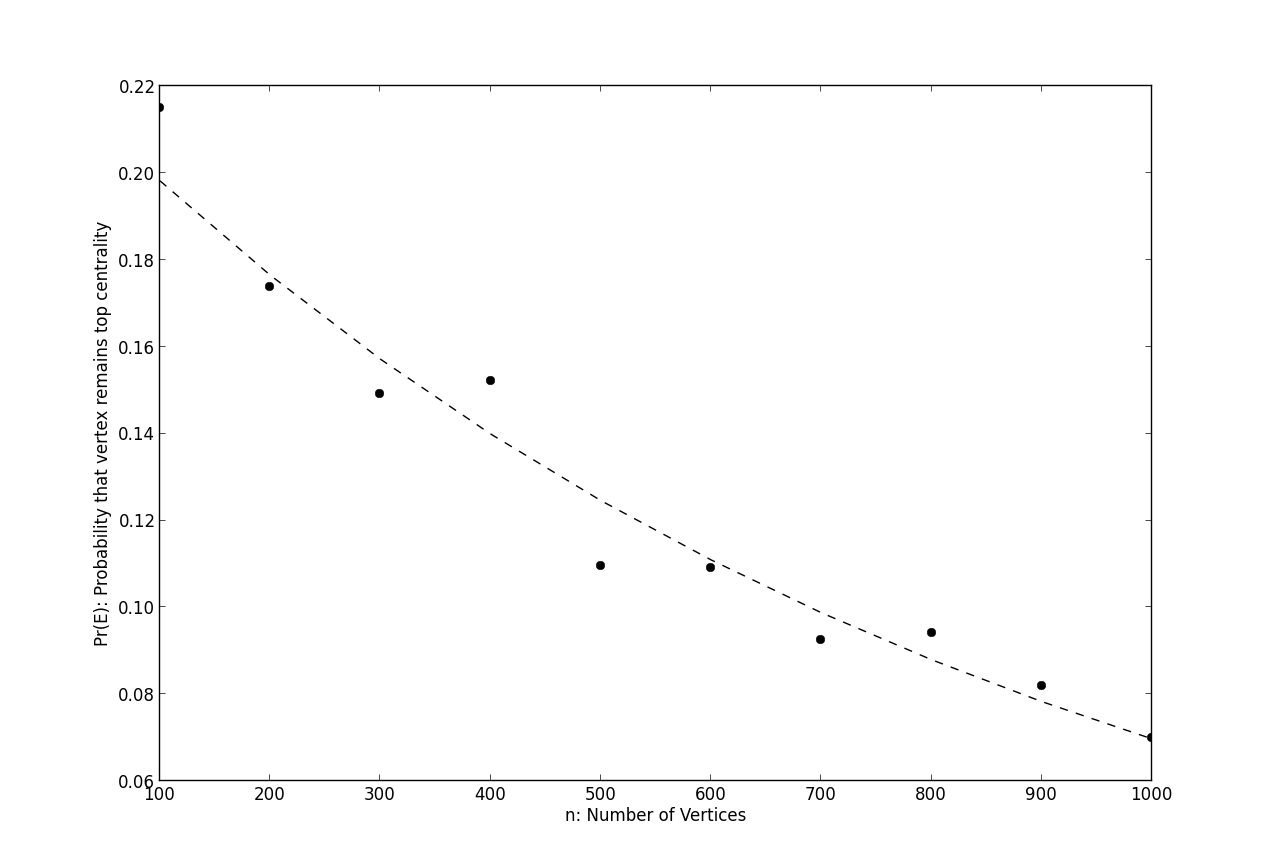}
    \captionof{figure}{The y-axis is the $Pr(E) = Pr(v^1_{max} = v^c_{max})$, where $p=q=0.025$. The best-fit line shows the exponential decay of the empirical simulations of $Pr(E)$ for increasing values of $n$. 10000 trials are made for each simulation point.}\label{fig:deg_centrality}
\end{center}

The degree of $v^1_{max}$ at $G^1$ is one of the factors that determines the probability that the node remains of highest centrality in the composite network. If the $deg_1(v^1_{max})$ is significantly larger then the rest of the vertices, then there is a higher probability that $v^1_{max}=v^c_{max}$. For example in the extreme case where $deg_1(v^1_{max}) = n-1$, it is not possible to find another vertex with degree greater than $n-1$ in $G^c$.

Since $deg_1(v^1_{max})$ is parameterized by $p$, Figure \ref{fig:deg_centrality_pq_ratio} shows the $Pr(E)$ for different ratios of $p$ and $q$, where without loss of generality $p > q$. As $p$ gets larger than $q$, the $Pr(E)$ increases. This is because for $p>q$, $G^1$ has more edges than $G^2$, hence $G^1$ dominates the behaviour of $G^c$.
\begin{center}
  \includegraphics[width=130mm]{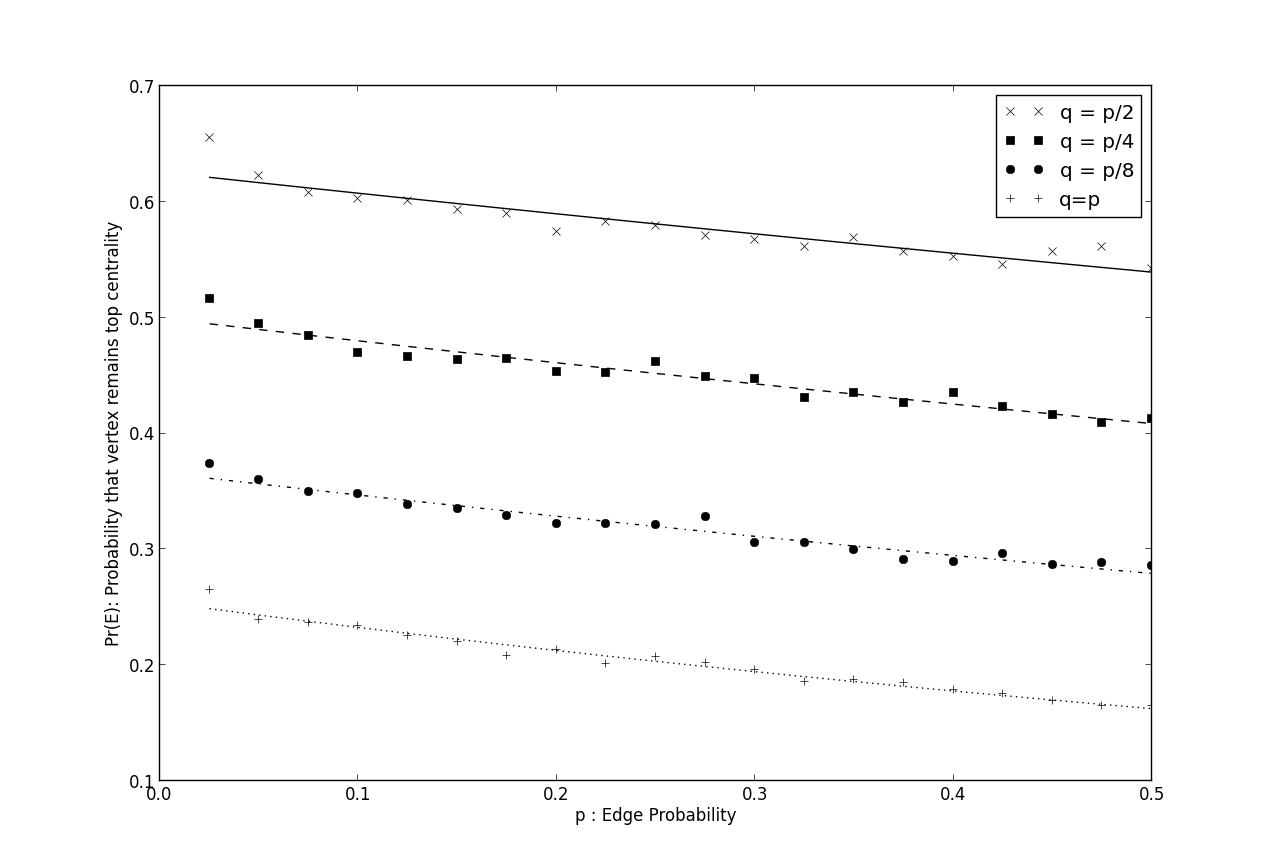}
    \captionof{figure}{$Pr(E) = Pr(v^1_{max} = v^c_{max})$ where $n=100$ and different $p:q$ ratios.}\label{fig:deg_centrality_pq_ratio}
\end{center}

\subsection{Asymptotic behaviour}\label{sec:deg_centrality_asy}
When one analyses the asymptotic behaviour of lemma \ref{deg_centrality} in the section above, it appears that the simulations in the section \ref{empirical} contradicts our analytical results. From the simulations, as $n$ gets larger, the probability decreases exponentially towards zero. However our analysis shows that it should converge to a nonzero limit. For example in Figure \ref{fig:deg_centrality}, the probability is 0.13 for $n=1000$. But for the same parameter, our analysis shows the convergence is to the limit $\approx 0.4739$!

The reasons lie in the details of Theorem \ref{thm:max_deg}. Firstly the theorem relies on Extreme Value Theory which is meaningful for the limits of very large values of $n$. 

Now we return to the issue that the asymptotic behaviour extrapolated  from the simulations differs from the analytical prediction. We believe this is to do with slow and non-monotonous convergence to the limit of large network sizes. We recall from Extreme Value Theory that the extreme of a set of $n$ normal distributed variables converges towards the Gumbel distribution logarithmically slow as $1/\ln(n)$ \cite{gumbel}. 

Since the degree distribution of the considered Erd\H{o}s-R\'{e}nyi networks is Gaussian, we accordingly expect very slow convergence towards the asymptotic result for the extremal degree. Moreover,  we conjecture that $Pr(E)$ depends non-monotonically on the number of nodes. The discrepancy between simulations and analytic results suggest that for increasing small-values of $n$, the probability initially decreases toward zero for then it increases to a non-zero limit. We imagine a behaviour similar to, say,  $y = (1 - 1.1^{-x})^{x}$.

From another perspective, the analytical result is intuitively understood by Order Statistics \cite{david}. The spacing between the $m$ largest and the $m-1$ largest  vertex degree approaches zero as $n \rightarrow \infty$ \cite{Mudholkar}. Hence we expect many vertices at the tail distribution to have degree approximately equal  to the maximum degree of the network. It is likely that $v^1_{max}$ belongs to the tail of the distribution of $G^c$, given that it belongs at the tail distribution of $G^1$. Thus it is   highly plausible that $v^1_{max}$ is also of maximum degree in $G^c$ and that the remaining the vertices have degrees belonging to the tail of the degree distribution of $G^c$.

%----Betweenness/Eigenvector Centrality----%
\section{Betweenness Centrality and Eigenvector Centrality}\label{sec:other_centrality}
Betweenness Centrality and Eigenvector Centrality are positively well correlated to the Degree Centrality of a network \cite{Valente}. This is because the mechanics of the two centrality measures favour vertices with high degree.

For example Eigenvector Centrality is the stable state of all vertices where every vertex's score is the sum of the centrality scores of its neighbours. Hence vertices with higher degree (more neighbours) have more components in the sum, and this result in higher Eigenvector Centrality score. 

Similarly the Betweenness Centrality of a vertex $v$ is the probability that the shortest path between a randomly chosen pair of vertices passes through $v$. A vertex of high degree has many edges leading to it and will therefore be more likely, than a vertex of low degree, to connect to a given shortest path between two arbitrarily chosen vertices.

The hypothesis that Betweenness/Eigenvector Centrality are positively correlated to the Degree Centrality is mainly supported by empirical verification \cite{Valente}. However the discrepancy we have found between our analytical and numerical results for the asymptotic behaviour of Degree Centrality suggests that caution is needed, and that the asymptotic limit may be difficult to reach through numerical simulations. So despite of the correlations expected between the different centrality measures, it may not be valid  to assume that the asymptotic behaviour of degree centrality  immediately also describes the asymptotic behaviour of  Betweenness and Eigenvector Centrality.

\subsection{Asymptotic behaviour}
Degree Centrality could be studied analytically because we could treat the vertex degree like an independent random variable. However the Betweenness/Eigenvector Centrality score of every vertex is a global aggregation of scores from the entire network, thus these vertex scores are dependent on each other and we are unable to compute the extremal behaviour by using independent random variables.

Recall from section \ref{sec:deg_centrality_asy} that the degree differences (spacing) of the vertices belonging to the tail of the distribution approaches zero for $n \rightarrow \infty$ \cite{Mudholkar}.  Hence there are \emph{many} vertices (for sufficiently large $n$) of degree close to the maximum degree in the network. We denote this set of top-percentile Degree Centrality vertices by $T$, and $v^1_{max} \in T$.

We repeat that a vertex's score of its Betweenness/Eigenvector Centrality is a global perspective, hence the local measure of  Degree Centrality is not sufficient to determine the order of the ranking of  Betweenness/Eigenvector Centrality of vertices.  However, we may treat the set of other measures like Clustering Coefficient as some perturbation to the Betweenness/Eigenvector Centrality ranking. This allows the equal degree vertices to have unequal centrality ranking ``randomly". When we apply this assumption to the set of vertices $T$ belonging to the tail of the distribution  the ranking of $v^1_{max}$ will be randomly ordered among the vertices in $T$.

This suggest that we estimate the probability that $v^1_{max}$ is the top in the set $T$ as $1/|T|$. Finally since $|T| \rightarrow \infty$ slowly as $n \rightarrow \infty$, we conclude that $\lim_{n\rightarrow \infty}Pr(E) \approx 1/|T| = 0$. This is different from the asymptotic behaviour of Degree Centrality, which converges to a nonzero limit.
 
\subsection{Empirical Results}
Fortunately for small $n$, the $Pr(E)$ for the three centrality measures are well correlated. So for small-values of $n$, we expect an increase in $n$ to lead to a decline in the probability. In addition the rate of convergence of Betweenness/Eigenvector Centrality will be \emph{much} slower than those of Degree Centrality. This is because the top Betweenness/Eigenvector centrality vertex often only requires high degree, not necessary the highest degree. Figure \ref{fig:bet_centrality} and Figure \ref{fig:eig_centrality} plot the simulations of $P(E)$ for the Betweenness Centrality and Eigenvector Centrality respectively.
\begin{center}
  \includegraphics[width=130mm]{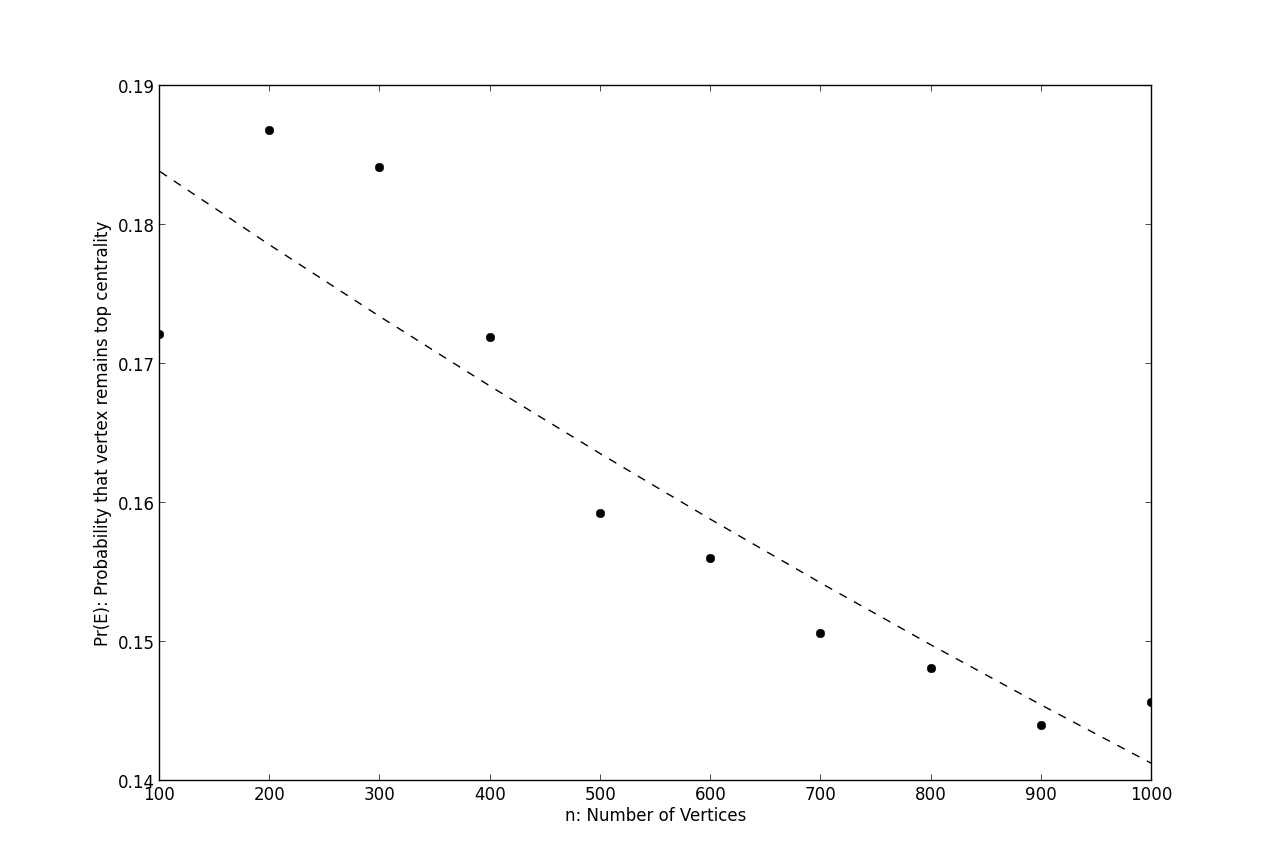}
    \captionof{figure}{The probability that the top Betweenness Centrality vertex of $G^1$ remains top in the composite network $G^1_{n,0.025} \cup G^2_{n,0.025}$. 10000 trials were made for each data point and the dotted line is the best-fit line.} \label{fig:bet_centrality}
\end{center}

\begin{center}
  \includegraphics[width=130mm]{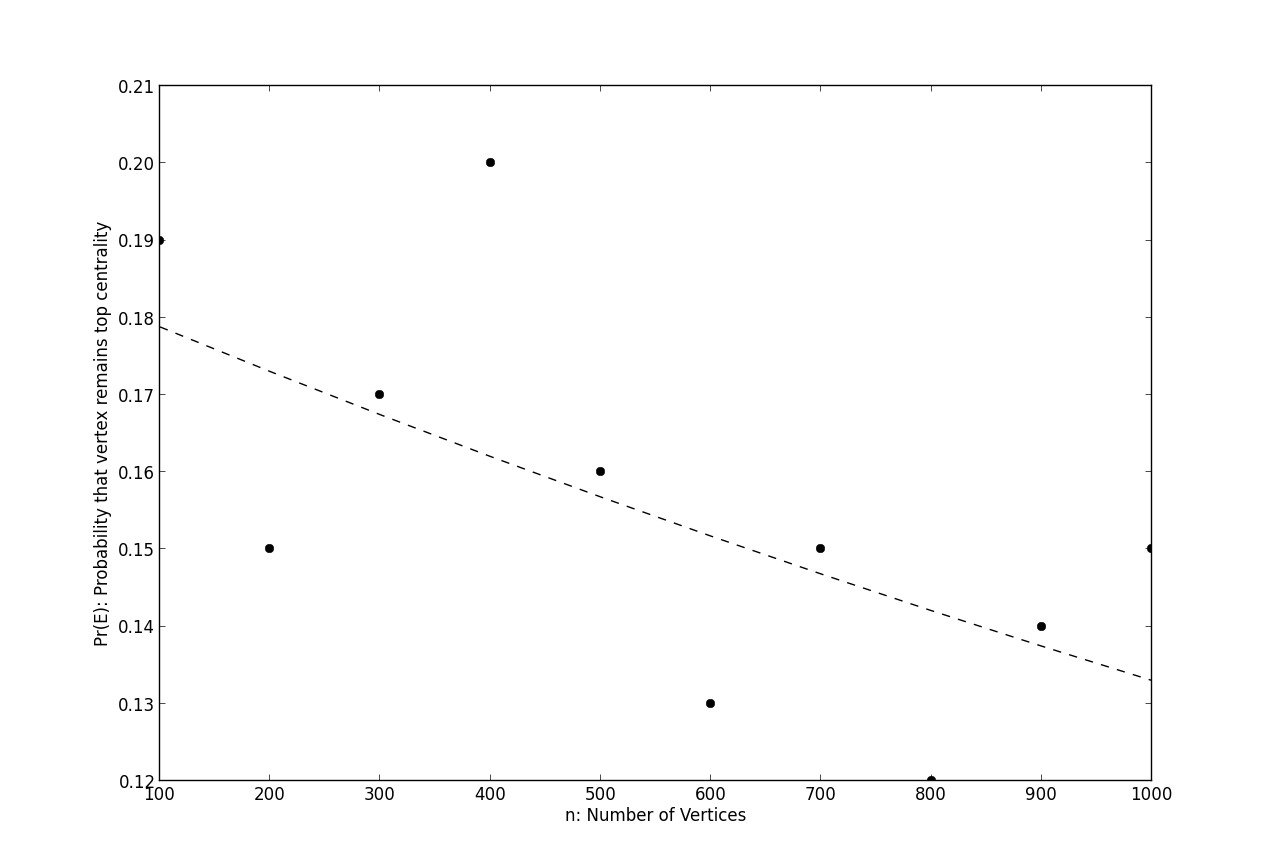}
    \captionof{figure}{The probability that the top Eigenvector Centrality vertex of $G^1$ remains top in the composite network $G^1_{n,0.025} \cup G^2_{n,0.025}$. 10000 trials were made for each data point and the dotted line is the best-fit line.} \label{fig:eig_centrality}
\end{center}

The rate  of decay for Eigenvector Centrality is much slower than Degree Centrality in Figure \ref{fig:deg_centrality}. For example simulations (not shown in figure) show that for Betweenness/Eigenvector Centrality, the $Pr(E) \approx 0.1$ at $n=5000$. Whereas for the Degree Centrality the same probability occurs at $n=700$.

\section{Discussion}
Degree Centrality, Betweenness Centrality and Eigenvector Centrality are some of the most common centrality metrics in Network Theory. This paper studies how these metrics change under the operation of edge union of networks on the same vertex set.

Under the edge union of Erd\H{o}s-R\'{e}nyi networks, the behaviour of the Degree Centrality is different from Eigenvector or Betweenness Centrality in the limits of large network. The analytical results helps to understand  the differences for the asymptotic behaviour of these metrics, which can be very hard to observe through simulations.

\subsection{Future Work}
The degree distribution of a  Erd\H{o}s-R\'{e}nyi Network follows a Poisson distribution, which is different from the power-law distribution found in many real-world networks. Hence a natural extension to this work is to consider the union between a scale-free network like Barab\'{a}si-Albert Network \cite{PhysRevE.65.057102} and a Erd\H{o}s-R\'{e}nyi Network. This result will be of relevance to the stability of centrality measures of real world networks \cite{Ghoshal,Costenbader}.
In this case the Erd\H{o}s-R\'{e}nyi Network (the $G^2$ network, abbreviated ER) can be perceived as a noisy perturbation on the interactions of the preferential attachment interactions in Barab\'{a}si-Albert Network (the $G^1$ network, abbreviated BA). Empirically the highest centrality vertices of BA is very stable, i.e. the probability that $v^1_{max}=v^c_{max}$ is high \cite{Ghoshal,Costenbader}. The construction of BA allows a minority of vertices to have degree a few orders of magnitude larger than the majority of the vertices.

What if this noisy perturbation follows a power-law distribution instead of a Poisson distribution? A hypothetical situation is when the ``noise" on the interactions is a confounding variable in the system that follows a power-law distribution. 

Figure \ref{fig:ba_union} compares the different stability behaviour of $BA\cup ER$ and $BA \cup BA$. The Degree Centrality of $BA \cup ER$ is extremely stable, where the highest centrality vertex remains top with 0.9 probability. In contrast $BA \cup BA$ shows a decline of $Pr(v^1_{max} = v^c_{max})$ for increasing network size.

Hence the centrality for BA will be significantly less stable if the noisy perturbation follows some preferential attachment. This could explain why the centrality of many real world static networks are inaccurate in predicting future outcomes. \cite{KimTAM12} showed that a random snapshot of a dynamic network loses some of the important relationships (edges), resulting in poor predictions. These missing edges might not be randomly chosen from a Gaussian model, but could possibly generated by mechanism similar to preferential attachment.

\begin{center}
  \includegraphics[width=130mm]{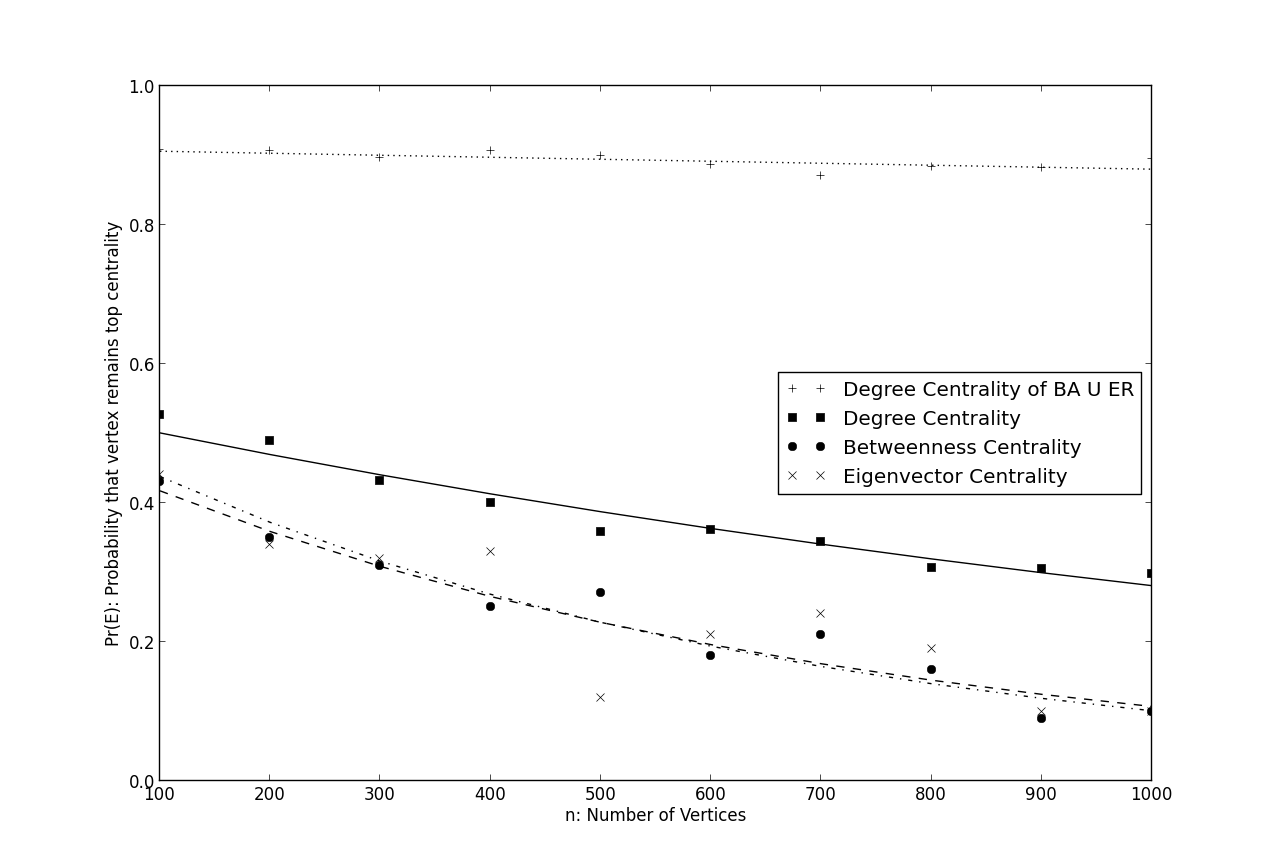}
    \captionof{figure}{The different centrality metrics of $BA \cup BA$ for increasing values of $n$ (network size). We left out the betweenness/eigenvector centrality of $BA \cup ER$ is because they are similar to the Degree Centrality of $BA \cup ER$.}\label{fig:ba_union}
\end{center}

\end{document}